\documentclass[12pt,english]{article}
\usepackage{lmodern}
\usepackage[T1]{fontenc}
\usepackage[latin9]{inputenc}
\usepackage{geometry}
\usepackage{color}
\usepackage{babel}
\usepackage{float}
\usepackage{mathtools}
\usepackage{url}
\usepackage{amsmath}
\usepackage{amsthm}
\usepackage{amssymb}
\usepackage{graphicx}
\usepackage{setspace}
\usepackage[authoryear]{natbib}
\usepackage[unicode=true,
 bookmarks=true,bookmarksnumbered=false,bookmarksopen=false,
 breaklinks=false,pdfborder={0 0 0},pdfborderstyle={},backref=false,colorlinks=true]
 {hyperref}
\hypersetup{pdftitle={Network Structure and Naive Sequential Learning},
 pdfpagelayout=OneColumn, pdfnewwindow=true, pdfstartview=XYZ, plainpages=false, urlcolor=[rgb]{0.0430 ,0, 0.5}, linkcolor=[rgb]{0.0430 ,0, 0.5}, citecolor=[rgb]{0.0430 ,0, 0.5}, hypertexnames=false}

\makeatletter

\providecommand{\tabularnewline}{\\}
\theoremstyle{remark}
\newtheorem{notation}{\protect\notationname}
\theoremstyle{definition}
\newtheorem{defn}{\protect\definitionname}
\theoremstyle{remark}
\newtheorem{rem}{\protect\remarkname}
\theoremstyle{plain}
\newtheorem{assumption}{\protect\assumptionname}
\theoremstyle{plain}
\newtheorem{lem}{\protect\lemmaname}
\theoremstyle{plain}
\newtheorem{prop}{\protect\propositionname}
\theoremstyle{definition}
 \newtheorem{example}{\protect\examplename}

\usepackage{chngcntr}
\setcitestyle{round}
\usepackage{mathtools}
\usepackage{breakcites}
\usepackage[all]{hypcap}
\usepackage{dcolumn}

\makeatother

\providecommand{\assumptionname}{Assumption}
\providecommand{\definitionname}{Definition}
\providecommand{\examplename}{Example}
\providecommand{\lemmaname}{Lemma}
\providecommand{\notationname}{Notation}
\providecommand{\propositionname}{Proposition}
\providecommand{\remarkname}{Remark}

\begin{document}
\title{Network Structure and Naive Sequential Learning\thanks{We thank Daron Acemoglu, J. Aislinn Bohren, Jetlir Duraj, Ben Enke,
Erik Eyster, Drew Fudenberg, Ben Golub, David Laibson, Jonathan Libgober,
Margaret Meyer, Pooya Molavi, Xiaosheng Mu, Matthew Rabin, Ran Spiegler,
Tomasz Strzalecki, Alireza Tahbaz-Salehi, Omer Tamuz, Linh T. Tô,
Muhamet Yildiz, and three anonymous referees for useful comments.}}
\author{Krishna Dasaratha\thanks{Harvard University. Email: \texttt{\protect\href{mailto:krishnadasaratha\%40gmail.com}{krishnadasaratha@gmail.com}}}
\and Kevin He\thanks{California Institute of Technology and University of Pennsylvania.
Email: \texttt{\protect\href{mailto:hesichao\%40gmail.com}{hesichao@gmail.com}}}}
\date{{\normalsize{}}%
\begin{tabular}{rl}
First version: & January 17, 2017\tabularnewline
This version: & February 29, 2020\tabularnewline
\end{tabular}}
\maketitle
\begin{abstract}
\begin{singlespace}
{\normalsize{}We study a sequential-learning model featuring a network
of naive agents with Gaussian information structures. Agents apply
a heuristic rule to aggregate predecessors' actions. They weigh these
actions according the strengths of their social connections to different
predecessors. We show this rule arises endogenously when agents wrongly
believe others act solely on private information and thus neglect
redundancies among observations. We provide a simple linear formula
expressing agents' actions in terms of network paths and use this
formula to characterize the set of networks where naive agents eventually
learn correctly. This characterization implies that, on all networks
where later agents observe more than one neighbor, there exist disproportionately
influential early agents who can cause herding on incorrect actions.
Going beyond existing social-learning results, we compute the probability
of such mislearning exactly. This allows us to compare likelihoods
of incorrect herding, and hence expected welfare losses, across network
structures. The probability of mislearning increases when link densities
are higher and when networks are more integrated. In partially segregated
networks, divergent early signals can lead to persistent disagreement
between groups.}{\normalsize\par}
\end{singlespace}

{\normalsize{}\thispagestyle{empty}
\setcounter{page}{0}}{\normalsize\par}
\end{abstract}
\interfootnotelinepenalty=10000

\newpage{}

\section{\label{sec:Introduction}Introduction}

Consider an environment with a sequence of agents facing the same
decision problem in turn, where each agent considers both her private
information and the behavior of those who came before her in reaching
a decision. When consumers choose between rival products, for instance,
their decisions are often informed by the choices of early customers.
When doctors decide on a treatment for their patients, they consult
best practices established by other clinicians who came before them.
Additionally, when a new theory or rumor is introduced into a society,
individuals are swayed by the discussions and opinions of those who
have already taken a clear stance on the new idea.

A key feature of these examples is that agents observe only the behavior
of a certain subset of their predecessors. For example, a consumer
may know about her friends' recent purchases, but not the product
choices of anyone outside her social circle. In general, each sequential
social-learning problem has an \emph{observation network} that determines
which predecessors are observable to each agent. Observation networks
associated with different learning problems may vary in density, extent
of segregation, and other structural properties. Hence, our central
research question: how does the structure of the observation network
affect the probability of correct social learning in the long run?

To answer this question, our model must capture key behavioral patterns
in how individuals process social information in these learning settings.
Empirical research on social learning suggests that humans often exhibit
\emph{inferential} \emph{naiveté}, failing to understand that their
predecessors' actions reflect a combination of private information
and the inference those predecessors have drawn from the behavior
of still others (e.g., \citet*{chandrasekhar2015testing,eyster2015experiment}).
Returning to the examples, a consumer may mistake a herd on a product
for evidence that everyone has positive private information about
the product's quality. In an online community, a few early opinion-makers
can make a rumor go viral, due to people not thinking through how
the vast majority of a viral story's proponents are just following
the herd and possess no private information about the rumor's veracity.

The present study examines the effect of the \emph{observation network}
on the extent of learning, in a setting where players suffer from
\emph{inferential naiveté}. We analyze the theoretical implications
of a tractable log-linear learning rule that aggregates observations
in a manner related to the DeGroot heuristic. Agents who fail to account
for correlations in their observations (as in \citealp{eyster2010naive})
choose actions according to this log-linear rule. We also introduce
\emph{weighted networks} and a weighted version of log-linear behavior,
where agents place different weights on different neighbors' actions.
We show that such decision weights arise when agents additionally
misperceive the precisions of predecessors' signals according to the
strengths of their links to said predecessors.

When combined with a Gaussian informational environment, our weighted
log-linear learning rule lets us compute naive agents' exact probability
of taking the correct action (``learning accuracy'') on arbitrary
weighted networks. In contrast, the existing literature on social
learning has focused on whether long-run learning outcome exhibits
certain properties (e.g., convergence to dogmatic beliefs in the wrong
state) with positive probability, but not on how these probabilities
vary across environments. In settings where learning is imperfect
(i.e., society does not almost surely learn the correct state in the
long run), we obtain a richer characterization of the learning outcome
and compute comparative statics of learning accuracy with respect
to network structures.

Under our weighted log-linear learning rule, actions take a simple
form: given a continuous action space and a binary state space, we
can express each agent's action as a log-linear function of her predecessors'
private signal realizations with coefficients depending on the weighted
network structure. We exploit this expression to develop a necessary
and sufficient condition for society to learn completely: no agent
has too much ``influence.'' Imperfect learning is the leading case:
some agent in the network must have disproportionate influence whenever
all but finitely many people observe more than one neighbor. Since
this condition applies to a very broad class of networks, our analysis
focuses on comparing differentially inefficient learning outcomes
across networks. The detailed comparative statics we obtain from this
approach are crucial: imperfect learning implies a wide range of welfare
losses on different networks.

Introducing naiveté generates predictions matching empirical observations
in several domains where the rational model's implications are unclear
or counterfactual. We prove that increasing network density leads
to more inaccurate social-learning outcomes in the long run. This
prediction is supported by the experimental results in our companion
paper \citep{dasaratha2019experiment}, where we find human subjects'
accuracy gain from social learning is twice as large on sparse networks
compared to dense networks. As another example, disagreement among
different communities is common in practice. In the domain of product
adoption, respective subcommunities frequently insist upon the superiority
of their preferred products. We prove that if agents' actions only
coarsely reflect their beliefs and society is partially segregated,
then two social subgroups can disagree forever.\footnote{Segregation is extensively documented in social networks (see \citealp*{mcpherson2001birds}
for a survey).} Because of the limited information conveyed by actions, disagreement
can persist even when agents observe the actions of unboundedly many
individuals from another group. This presents a sharp contrast with
Bayesian social-learning models (and other leading learning models
such as DeGroot learning), where asymptotic agreement is a robust
prediction.

\subsection{Related Literature}

\subsubsection{Effect of Network Structure on Learning}

Much of the literature on how network structure matters for learning
outcomes has focused on networked agents repeatedly guessing a state
while learning from the same set of neighbors each period (e.g., \citet{bala1998learning}).

The leading behavioral model here is the DeGroot heuristic, which
forms a belief in each period by averaging the beliefs of neighbors
in the previous period. A key prediction of DeGroot learning is that
society converges to a long-run consensus \citep*{demarzo2003persuasion},
which will be correct in large networks as long as the network is
not too unbalanced \citep{golub2010naive}. Much of the analysis focuses
on how network structure (e.g., homophily) matters for the speed of
convergence to correct consensus beliefs \citep{golub2012homophily}.
We find that in a sequential setting, natural changes in network structure
matter for asymptotic accuracy, not only for speed of learning. Changing
network density, which has no effect on DeGroot learning in large
networks, can substantially alter the probability that a society learns
correctly. One intuition for this difference is that DeGroot agents
assign weights adding up to 1 to their neighbors, but agents in our
setting have increasing out-degrees with increasing network density
and therefore can overweight their social information. Homophily also
matters for this probability and even for whether consensus is ever
reached.

While DeGroot proposes the averaging rule as an ad hoc heuristic,
several recent papers have developed behavioral microfoundations for
learning in the repeated-interaction setting \citep*{molavi16,mueller2019general,levy2017information}.
These models closely resemble ours at the level of individual behavior,
but their predictions about society's long-run beliefs are more in
line with DeGroot. As such, changes in network structure again have
a limited scope for affecting learning outcomes in this literature.

\subsubsection{Sequential Social Learning}

We consider the same environment as the extensive literature on sequential
social learning beginning with \citet{banerjee1992simple} and \citet*{bikhchandani1992theory}.
\citet*{acemoglu2011bayesian} and \citet{lobel2015information} characterize
network features that lead to correct asymptotic learning for Bayesians
who move sequentially. By providing a thorough understanding of rational
learning in sequential settings, this literature provides a valuable
benchmark as we study naive learning. We find that among network structures
where Bayesian agents learn asymptotically, there is large variation
in the probability of mislearning for naive agents.

Several authors look at sequential behavioral learning on a particular
network structure, usually the complete network \citep{eyster2010naive,bohren2016informational,bohren2017bounded}.
We characterize several ways in which the choice of network structure
matters for the distribution of long-run outcomes. \citet{eyster2014extensive}
exhibit a general class of social-learning rules, which includes the
weighted log-linear rule we study for certain values of the weights,
where mislearning occurs with positive probability. We go beyond this
general result by deriving expressions for the exact probabilities
of mislearning on different networks, whose associated welfare losses
cannot be compared using the binary classification of \citet{eyster2014extensive}.

\section{Model\label{sec:Model}}

\subsection{Sequential Social Learning on a Weighted Network}

There are two possible states of the world, $\omega\in\{0,1\}$, both
equally likely. There is an infinite sequence of agents indexed by
$i\in\mathbf{\mathbb{N}}$. Agents move in order, each acting once.

On her turn, agent $i$ observes a private signal $s_{i}\in\mathbb{R}$.
Private signals $(s_{i})$ are Gaussian and independent and identically
distributed(i.i.d.) given the state. When $\omega=1$, we have $s_{i}\sim\mathcal{N}(1,\sigma^{2})$
for some conditional variance $\sigma^{2}>0$. When $\omega=0$, we
have $s_{i}\sim\mathcal{N}(-1,\sigma^{2})$.

In addition to her private signal, agent $i$ also observes the actions
of previous agents. Then, $i$ chooses an action $a_{i}\in[0,1]$.
In our microfoundation for Definition \ref{def:log_linear}, agent
$i$ chooses $a_{i}$ to maximize the expectation of 
\[
u_{i}(a_{i},\omega)\coloneqq-(a_{i}-\omega)^{2}
\]
given her belief about $\omega$, which we describe later. So her
chosen action corresponds to the probability she assigns to the event
$\{\omega=1\}$.

We find it convenient to work with the following change of variables.
\begin{notation}
We have $\tilde{s}_{i}\coloneqq\ln\left(\frac{\mathbb{P}[\omega=1|s_{i}]}{\mathbb{P}[\omega=0|s_{i}]}\right)$
and $\tilde{a}_{i}\coloneqq\ln\left(\frac{a_{i}}{1-a_{i}}\right)$
\end{notation}
In words, $\tilde{s}_{i}$ is the log-likelihood ratio of the events
$\{\omega=1\}$ and $\{\omega=0\}$ given signal $s_{i}$, while it
is easy to show that $\tilde{a}_{i}$ is the log-likelihood ratio
of $\{\omega=1\}$ and $\{\omega=0\}$ corresponding to action $a_{i}$.
That is to say, if $a_{i}$ is optimal given $i$'s beliefs, then
$\tilde{a}_{i}$ is the log-likelihood ratio of $\{\omega=1\}$ and
$\{\omega=0\}$ according to $i$'s beliefs. Note that the transformations
from $s_{i}$ to $\tilde{s}_{i}$ and from $a_{i}$ to $\tilde{a}_{i}$
are bijective, so no information is lost when we relabel variables.

Agents are linked to all of their predecessors on a weighted network,
with a lower-triangular adjacency matrix $M$ where all diagonal entries
are equal to 0. For $i>j,$ the weight of the link from $i$ to $j$
is given by $M_{i,j}\in[0,1]$. The weights of the edges determine
the relative importance agents place on others' actions in forming
their beliefs. In Section \ref{sec:Applications}, we derive comparative
statics with respect to the network structure. Because studying continuous
changes in the network is more tractable than discrete changes, we
consider a model that allows interior network weights.

Throughout, we study \emph{naive agents} who choose actions equal
to a weighted sum of their observations according to the following
log-linear updating rule.
\begin{defn}
\label{def:log_linear} Agents use the \emph{weighted log-linear rule
}if each agent $i$ plays 
\[
\tilde{a}_{i}=\tilde{s}_{i}+\sum_{j<i}M_{i,j}\tilde{a}_{j}.
\]
\end{defn}
In words, each agent $i$'s log action is a weighted sum of her predecessors'
log actions and her own log signal. The network $M$ exogenously determines
the relative influences of different predecessors' behavior on $i$'s
play, with $j$'s influence proportional to the strength of $i$'s
social connection to her. By contrast, a society of rational agents
would put endogenous weights on others' actions that are not simply
proportional to strengths of the network links between them. For example,
if agents only observe the actions of linked neighbors, rational agents
would play the unique perfect Bayesian equilibrium of the social-learning
game, in which case some equilibrium decision weights may be negative.\footnote{We omit the proof that actions are log linear at the perfect Bayesian
equilibrium. This can be shown by induction, and the key step is a
calculation similar to Lemma \ref{lem:llh_normal}.}
\begin{rem}
\label{rem:DeGrootEmbedding}The formula in Definition \ref{def:log_linear}
resembles the DeGroot updating rule. A key distinction is that we
allow for agents to have any out-degree, while the DeGroot heuristic
requires all agents' weights sum to 1. In an unweighted network, any
agent with multiple observations has an out-degree greater than 1.
This distinction is not just a normalization, but is, in fact, the
source of redundancy under naive inference.
\end{rem}

\subsection{\label{subsec:Naive-Inference-Assumption}Microfoundation for Weighted
Log-Linear Rule}

In this subsection, we provide a psychological microfoundation for
the weighted log-linear rule from Definition \ref{def:log_linear}.
We first show that on \emph{unweighted networks} (i.e., when each
$M_{i,j}$ is either 0 or 1), this rule follows from a primitive assumption
about agents' inference.

A growing body of recent evidence in psychology and economics shows
that people learning from peers are often not fully correct in their
treatment of social structure \citep*{chandrasekhar2015testing,enke2016correlation}.
Instead of calculating the optimal Bayesian behavior that fully takes
into account all they know about the network and signal structure,
agents often apply heuristic simplifications to their environment.
When networks are complicated and/or uncertain, determining Bayesian
behavior can be intractable \citep*{jadbabaie2017bayesian} and these
heuristic learning rules become especially prevalent \citep*{eyster2015experiment}.
Motivated by this observation, we consider the following behavioral
assumption.
\begin{assumption}
\label{assu:behavioral}Each agent wrongly believes that each predecessor
chooses an action to maximize her expected payoff based only on her
private signal, and not on her observation of other agents.
\end{assumption}
This inferential mistake can be equivalently described as agent $i$
misperceiving $M_{j,k}=0$ for all $k<j<i$. Under this interpretation,
$i$ acts as if her neighbors do not take into account their own predecessors'
actions.

In the sequential-learning literature, Assumption \ref{assu:behavioral}
was first studied on the complete network by \citet{eyster2010naive},
who coined the term ``best-response trailing naive inference'' (BRTNI)
to describe this behavior. The laboratory games in \citet*{eyster2015experiment}
and \citet*{mueller2015general} find evidence for this behavioral
assumption.

Agents who make this inferential mistake use the log-linear rule on
unweighted networks, providing a psychological microfoundation for
the behavior we study.
\begin{lem}
\label{lem:On-an-unweighted}On an unweighted network where agents
observe only the actions of linked predecessors, Assumption \ref{assu:behavioral}
implies agents use the weighted log-linear rule.
\end{lem}
Due to the inferential mistake, agent $i$ wrongly infers that $j$'s
log action equals her log signal. (This inference is possible since
the continuum action set is rich enough to exactly reveal beliefs
of predecessors.) The action $a_{i}$ is the product of the relevant
likelihoods because an agent satisfying Assumption \ref{assu:behavioral}
thinks her observations are based on independent information, and
therefore $\tilde{a}_{i}$ is the sum of the corresponding log-likelihood
ratios.
\begin{rem}
Inference under Assumption \ref{assu:behavioral} is cognitively simple
in that it does not rely on agents' knowledge about the network (beyond
their own neighborhoods) or even knowledge about the order in which
their predecessors moved. Our model therefore applies even to complex
environments with random arrival of agents. In such environments,
Assumption \ref{assu:behavioral} may be more realistic than assuming
full knowledge about the observation structure and move order.
\end{rem}
Next we give a microfoundation for the same behavior on weighted networks.
We provide an interpretation of network weights that formalizes the
idea that agents place more trust in neighbors with whom their connections
are stronger: we suppose that agents underestimate the precision of
others' private signals in a way that depends on $M_{i,j}$.
\begin{assumption}
\label{assu:wrong_precision} Given network weight $M_{i,j}\in[0,1],$
agent $i$ believes $j$'s private signal has conditional variance
$\sigma^{2}/M_{i,j}$ given the state.
\end{assumption}
\citet{weizsacker2010we}'s meta-analysis of sequential social-learning
experiments finds that laboratory subjects underuse social information
relative to their own private signals. Our Assumption \ref{assu:wrong_precision}
is consistent with this evidence, but also allows for different degrees
of underuse for different predecessors. Weaker network connections
formally correspond to predecessors whose signals are believed to
be less informative about the state or less relevant. Conversely,
if we know that $i$ acts as if $j$'s signal has conditional variance
$V_{i,j}\ge\sigma^{2}$, then we can construct a weighted network
with weights $M_{i,j}=\sigma^{2}/V_{i,j}$.

The next result shows the combination of the inferential mistake about
others' social information and the underestimation of others' signal
precisions (Assumptions \ref{assu:behavioral} and \ref{assu:wrong_precision})
provides a microfoundation for the weighted log-linear rule.
\begin{lem}
\label{lem:precision}Agents who satisfy Assumptions \ref{assu:behavioral}
and \ref{assu:wrong_precision} use the weighted log-linear rule.
\end{lem}
By a property of the Gaussian distribution, a log-transformed Gaussian
variable is also Gaussian, which is the key to showing the above lemma.

\subsection{Complete Learning and Mislearning}

We define what it means for society to learn completely in terms of
convergence of actions.
\begin{defn}
\emph{\label{def:Society-learns-correctly}}Society\emph{ learns completely}
if $(a_{n})$ converges to $\omega$ in probability.
\end{defn}
Since $a_{n}$ reflects agent $n$'s belief in $\{\omega=1\}$ in
our microfoundation of weighted log-linear inference, this definition
also describes a property about the convergence of beliefs.\footnote{In general, we treat $a_{n}$ as the belief of a naive agent who plays
$\tilde{a}_{n}$.} In a setting where society learns completely, agent $n$ becomes
very likely to believe strongly in the true state of the world as
$n$ grows large.

One failure of complete learning is when society becomes fully convinced
of the wrong state of the world with positive probability, an event
we call mislearning.
\begin{defn}
Society \emph{mislearns} when $\lim_{n\to\infty}a_{n}=0$ but $\omega=1$
or when $\lim_{n\to\infty}a_{n}=1$ but $\omega=0$.
\end{defn}
\begin{rem}
\label{rem:non-convergence} Mislearning is not the only obstacle
to complete learning. Consider a network where, for $i\ge3,$ we have
$M_{i,1}=M_{i,2}=1$ and $M_{i,j}=0$ for all $j\ne1,2$. Clearly
this society neither learns completely nor mislearns with positive
probability. Instead, agents' beliefs almost surely do not converge.
\end{rem}

\section{Social Influence and Learning\label{sec:Results}}

In this section, we develop a necessary and sufficient condition on
the network for society to mislearn. We argue that this condition
is satisfied by a large class of networks of economic relevance.

\subsection{\label{subsec:Path-Counting-Interpretation-of}Path-Counting Interpretation
of Actions}

We now show that with naive agents, actions have a simple (log-)linear
expression in terms of paths in the network. Unlike the expression
in Definition \ref{def:log_linear}, this next result expresses actions
in terms of only signal realizations and the network structure, making
no reference to predecessors' actions.

Let $M[n]$ refer to the $n\times n$ upper-left submatrix of $M$.
\begin{prop}
\label{prop:representation}Consider any weighted network $M$. For
each $n$, the actions of the first $n$ agents are determined by
\[
\left(\begin{array}{c}
\tilde{a}_{1}\\
\vdots\\
\tilde{a}_{n}
\end{array}\right)=(I-M[n])^{-1}\cdot\left(\begin{array}{c}
\tilde{s}_{1}\\
\vdots\\
\tilde{s}_{n}
\end{array}\right).
\]
 So, $\tilde{a}_{i}$ is a linear combination of $(\tilde{s}_{j})_{j=1}^{i}$,
with coefficients given by the number of weighted paths from $i$
to $j$ in the network with adjacency matrix $M$.
\end{prop}
From a combinatorial perspective, the formula says that the influence
of $j$'s signal on $i$'s action depends on the number of weighted
paths from $i$ to $j$. In unweighted networks where all entries
in $M$ are 1 or 0, this is just the number of paths. In general,
``weighted paths'' means the path passing through agents $i_{0},...,i_{K}$
is counted with weight $\prod_{k=0}^{K-1}M_{i_{k},i_{k+1}}$.

Our Proposition \ref{prop:representation} resembles a formula for
agents' actions in \citet{levy2017information}. In a setting of repeated
interaction with a fixed set of neighbors instead of sequential social
learning, \citet{levy2017information} also find that the influence
of $i$'s private information on $j$'s period $t$ posterior belief
depends on the number of length-$t$ paths from $i$ to $j$ in the
network.

\subsection{Condition for Complete Learning}

We now use the representation result of Proposition \ref{prop:representation}
to study which networks lead to mislearning by naive agents.

We define below a notion of network influence for the sequential social-learning
environment, which plays a central role in determining whether society
learns completely in the long run.
\begin{defn}
Let $b_{i,j}\coloneqq(I-M[i])_{i,j}^{-1}$ be the number of weighted
paths from $i$ to $j$ in network $M$.
\end{defn}
Because of Proposition \ref{prop:representation}, these path counts
are important to our analysis.
\begin{defn}
For $n>i,$ the \emph{influence} of $i$ on $n$ is $\mathbb{I}(i\to n)\coloneqq b_{n,i}/\sum_{j=1}^{n}b_{n,j}$.
\end{defn}
That is to say, the influence of $i$ on $n$ is the fraction of paths
from $n$ that end at $i$.

A different definition of influence appears in \citet{golub2010naive},
who study DeGroot learning in a network where agents act simultaneously
each period. For them, the influence of an agent $i$ is determined
by the unit left eigenvector of the belief-updating matrix, which
is proportional to $i$'s degree in an undirected network with symmetric
weights. Both definitions are related to the proportion of walks terminating
at an agent, but because of the asymmetry between earlier and later
agents in the sequential setting, the distribution of walks tends
to be more unbalanced.

For Proposition \ref{prop:influence_correct_learning} only, we consider
networks that satisfy the following connectedness condition.
\begin{defn}
Network $M$ satisfies the \emph{connectedness condition} if there
is an integer $N$ and constant $C>0$ such that for all $i>N$, there
exists $j<N$ with $b_{i,j}\ge C$.
\end{defn}
Intuitively, this says that all sufficiently late agents are indirectly
influenced by some early agent. An unweighted network satisfies the
connectedness condition if and only if there are only finitely many
agents who have no neighbors. If such a network violates the connectedness
condition, then clearly the infinitely many agents without neighbors
will prevent society from learning completely. All weighted networks
studied in Section \ref{sec:Applications} also satisfy the connectedness
condition.
\begin{prop}
\label{prop:influence_correct_learning}Consider any weighted network
satisfying the connectedness condition. Society learns completely
if and only if $\lim_{n\to\infty}\mathbb{I}(i\to n)=0$ for all i.
\end{prop}
Proposition \ref{prop:influence_correct_learning} says that beliefs
always converge to the truth if and only if no agent has undue influence
in the network. This is a recurring insight in research on social
learning on networks, beginning with the ``royal family'' example
and related results in \citet{bala1998learning}. Other examples where
excessive influence hinders social learning in networks include \citet{golub2010naive},
\citet*{acemoglu2010spread}, and \citet*{mossel2015strategic}. The
main contribution of Proposition \ref{prop:influence_correct_learning}
is to identify the relevant measure of influence in our sequential-learning
setting with naive agents. In our setting, unlike on large unordered
networks as in \citet{golub2010naive}, the ordering of agents creates
an asymmetry that prevents society from learning completely on most
natural networks. Early movers influence many successors, an asymmetry
unique to the sequential-learning setting. For instance, the results
of \citet{golub2010naive} imply that when agents move simultaneously
and every agent weights every other agent equally, society converges
to complete learning as the size of the network grows. But as we show
in Section \ref{subsec:Uniform-Weights}, society does not learn completely
in the uniform weighted network where each agent is connected to every
predecessor with the same weight.

The idea behind the proof is that if there were some $i$ and $\epsilon>0$
such that $\mathbb{I}(i\to n)>\epsilon$ for infinitely many $n$,
then $i$ exerts at least $\epsilon$ influence on all these future
players. Since $\tilde{s}_{i}$ is unbounded, there is a rare but
positive probability event where $i$ gets such a strong but wrong
private signal so that any future player who puts $\epsilon$ weight
on $\tilde{s}_{i}$ and $(1-\epsilon)$ weight on other signals would
come to believe in the wrong state of the world with high probability.
But this would mean infinitely many players have a high probability
of believing in the wrong state of the world, so society fails to
learn completely. To gain an intuition for the converse, first observe
that $\tilde{a}_{n}=\|\vec{b}_{n}\|_{1}\sum_{i=1}^{n}\mathbb{I}(i\to n)\tilde{s}_{i}$.
In the event that $\omega=1$, the mean of $\tilde{a}_{n}$ converges
to infinity with $n$. So, provided the variance of $\tilde{a}_{n}$
is small relative to its mean, $\tilde{a}_{n}$ will converge to infinity
in probability and society will learn completely. Since the log signals
$(\tilde{s}_{i})$ are i.i.d$.$, the variance of $\tilde{a}_{n}$
is small relative to its mean precisely when all of the weights $\mathbb{I}(i\to n)$
in the summand are small; this is guaranteed by the condition on influence
$\lim_{n\to\infty}\mathbb{I}(i\to n)=0$.

We now argue, both analytically and through an example, that the condition
for complete learning in Proposition \ref{prop:influence_correct_learning}
is violated by a large class of weighted networks. The \emph{out-degree}
of an agent $i$ is $\sum_{j<i}M_{i,j}$, interpreted as the total
number of neighbors who directly affect $i$'s play. We first show
that on any network where all but finitely many agents have out-degree
at least $1+\epsilon$ for some $\epsilon>0$, complete learning fails.
\begin{prop}
\label{prop:one_plus_epsilon}Suppose there exists $\epsilon>0$ so
that $\sum_{j<i}M_{i,j}\ge1+\epsilon$ for all except finitely many
agents $i.$ Then society does not learn completely.
\end{prop}
The proof establishes that such a network satisfies the connectedness
condition, but the influence of at least one of the early agents does
not converge to 0, so complete learning fails by Proposition \ref{prop:influence_correct_learning}.
The intuition is that if an influential later agent has an out-degree
greater than 1, then the earlier agents whose action indirectly affects
her must have even more influence. We construct a correspondence between
weighted paths ending at early agents and weighted paths ending at
later agents.

The condition in Proposition \ref{prop:one_plus_epsilon} is satisfied
by all of the weighted networks studied in Section \ref{sec:Applications},
which all feature mislearning with positive probability. The network
in Remark \ref{rem:non-convergence} also satisfies the condition
in Proposition \ref{prop:one_plus_epsilon} and almost surely leads
to nonconvergence of beliefs.

As an additional example, consider a network where network weights
decay exponentially in distance, so $M_{i,j}=\delta^{i-j}$ for some
$\delta\ge0$. When the rate of decay is strictly above the threshold
of $\frac{1}{2}$, late enough agents have out-degree bounded away
from 1, so society does not learn completely by Proposition \ref{prop:one_plus_epsilon}.
When the rate of decay is strictly below the same threshold, we can
show the connectedness condition fails and agents' beliefs do not
converge to $\omega$ due to lack of information. At the threshold
value of $\delta=\frac{1}{2}$, private signals of all predecessors
are given equal weight, so the law of large numbers implies complete
learning. This highlights the fragility of complete learning in our
model.
\begin{example}
\label{exa:decay}Suppose $M_{i,j}=\delta^{i-j}$ for some $\delta\ge0$.
Society learns completely if and only if $\delta=\frac{1}{2}$.
\end{example}
Details of the arguments are provided in the Appendix.

\section{Probability of Mislearning and Network Structure\label{sec:Applications}}

In this section, we compare the probability of mislearning in networks
where complete learning fails by Proposition \ref{prop:influence_correct_learning}.
To do so, we first derive a formula for the probability of mislearning
as a function of the observation network. Then, applying this expression
to several canonical network structures, we compute comparative statics
of this probability with respect to network parameters.

The first network structure we consider assigns the same weight to
each link. Next, we study a homophilic network structure with agents
split into two groups, allowing different weights on links within
groups and between groups.

\subsection{Probability of Mislearning}

Due to the Gaussian signal structure, we can give explicit expressions
for the distributions of agent actions in each period. We show that
the probability that agent $n$ is correct about the state is related
to the ratio of $\ell_{1}$ norm to $\ell_{2}$ norm of the vector
of weighted path counts to $n$'s predecessors, $\vec{b}_{n}\coloneqq(b_{n,1},...,b_{n,n})$.
The ratio $\|\vec{b}_{n}\|_{1}/\|\vec{b}_{n}\|_{2}$ can be viewed
as a measure of distributional equality for the vector of weights
$\vec{b}_{n}$.\footnote{In fact, the ratio of $\ell_{1}$ to $\ell_{2}$ norm has been used
in the applied mathematics literature as a measure of normalized sparsity.} Indeed, among positive $n$-dimensional vectors $\vec{b}_{n}$ with
$\|\vec{b}_{n}\|_{1}=1,$ the $\ell_{1}/\ell_{2}$ ratio is minimized
by the vector $\vec{b}_{n}=(1,0,...,0)$ and maximized by the vector
$\vec{b}_{n}=(\frac{1}{n},\frac{1}{n},...,\frac{1}{n})$.

We can express in terms of the network structure the ex ante probability
that agent $n$ puts more confidence in the state being $\omega=1$
when this is, in fact, true. This gives the key result that later
lets us compare the probabilities of mislearning on different networks.
\begin{prop}
\label{prop:prob_mislearn} On any network, the probability that agent
$n$ thinks the correct state is more likely than the incorrect one
is 
\[
\mathbb{P}[\tilde{a}_{n}>0\mid\omega=1]=\Phi\left(\frac{1}{\sigma}\cdot\frac{\parallel\vec{b}_{n}\parallel_{1}}{\parallel\vec{b}_{n}\parallel_{2}}\right),
\]
where $\Phi$ is the standard Gaussian distribution function.
\end{prop}
As $\frac{\parallel\vec{b}_{n}\parallel_{1}}{\parallel\vec{b}_{n}\parallel_{2}}$
increases, the probability of agent $n$ playing higher actions in
state $\omega=1$ also increases. In other words, the agent is more
likely to be correct about the state when the vector of path counts
is more evenly distributed. This should make intuitive sense as she
is more likely to be correct when her action is the average of many
independent signals with roughly equal weights, and less likely to
be correct when her action puts disproportionally heavy weights on
a few signals.

The proof of Proposition \ref{prop:prob_mislearn} first expresses
$\tilde{a}_{n}=\sum_{i=1}^{n}b_{n,i}\tilde{s}_{i}$ using Proposition
\ref{prop:representation} and then observes that $(s_{i})$ are distributed
i.i.d$.$ $\mathcal{N}(1,\sigma^{2})$ conditional on $\omega=1$.
This means $(\tilde{s}_{i})$ are also conditionally i.i.d$.$ Gaussian
random variables, since the proof of Lemma \ref{lem:precision} establishes
that $\tilde{s}_{i}=2s_{i}/\sigma^{2}.$ As a sum of conditionally
i.i.d$.$ Gaussian random variables, the action $\tilde{a}_{n}$ is
itself Gaussian. The result follows from calculating the mean and
variance of this sum.

For the remainder of this section, we study specific weighted networks
where the ratio $\|\vec{b}_{n}\|_{1}/\|\vec{b}_{n}\|_{2}$ can be
expressed in terms of interpretable network parameters. Our basic
technique is to count paths on a given network using an appropriate
recurrence relation, and then to apply Proposition \ref{prop:prob_mislearn}.
This allows us to relate network parameters to the probability distribution
over learning outcomes.

\subsection{Uniform Weights\label{subsec:Uniform-Weights}}

The simplest network we consider assigns the same weight $q\in[0,1]$
to each feasible link. By varying the value of $q$, we can ask how
link density affects the probability of mislearning, which we now
define.
\begin{prop}
\label{prop:Uniform-Weights}Consider the $q$-uniform weighted network.
When $0<q\le1$, almost surely agents' actions $a_{n}$ converge to
$0$ or $1$. The probability that society mislearns is 
\[
\Phi\left(-\frac{1}{\sigma}\cdot\sqrt{\frac{q+2}{q}}\right).
\]
This probability is strictly increasing in $q$.
\end{prop}
The first statement of the proposition tells us that agents eventually
agree on the state of the world, and that these beliefs are arbitrarily
strong after some time. These consensus beliefs need not be correct,
however. The probability of society converging to incorrect beliefs
is nonzero for all positive $q$, and increases in $q.$ When the
observational network is more densely connected, society is more likely
to be wrong, as in Figure \ref{fig:uniform_figure}.

\begin{figure}[H]
\begin{centering}
\includegraphics[scale=0.7]{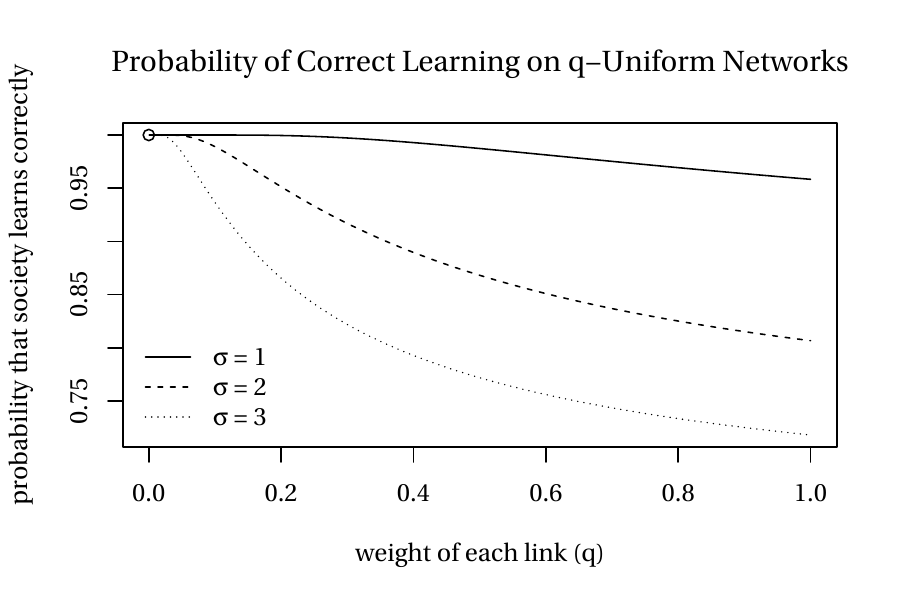}
\par\end{centering}
\caption{\label{fig:uniform_figure}Probability that society learns correctly
($a_{n}\to\omega)$ on a network where all feasible links have weight
$q$. See Remark \ref{rem:discontinuous} for the discontinuity at
$q=0$.}
\end{figure}

When the observation network is sparse (i.e., $q$ is low), early
agents' actions convey a large amount of independent information because
they do not influence each other too much. This facilitates later
agents' learning. For high $q$, early agents' actions are highly
correlated, so later naive agents cannot recover the true state as
easily. A related intuition compares agents' beliefs about network
structure to the actual network: as $q$ grows, agents' beliefs about
the network weights chosen by their neighbors differ more and more
from the true weights. For small $q$, however, underweighting of
social information partially mitigates the error due to Assumption
\ref{assu:behavioral}. To complement this theoretical result, in
a companion paper we conduct a sequential-learning experiment to evaluate
a related comparative static \citep{dasaratha2019experiment}. In
line with the intuition above, we find that human subjects indeed
exhibit lower long-run accuracy in the learning game when the density
of the observation network increases.

The proof relies on the recurrence relation $b_{n,i}=(1+q)b_{n-1,i}$.
To see that this recurrence holds, let $\Psi[n\rightarrow i]$ be
the set of all paths from $n$ to $i$ and let $\Psi[(n-1)\rightarrow i]$
be the set of all paths from $n-1$ to $i$. For each $\psi\in\Psi[(n-1)\rightarrow i]$
passing through agents $(n-1),j_{1},j_{2},...,i$, we associate two
paths $\psi^{'},\psi^{''}\in\Psi[n\rightarrow i]$, with $\psi^{'}$
passing through $n,j_{1},j_{2},...i$ and $\psi^{''}$ passing through
$n,(n-1),j_{1},j_{2},...,i$. This association exhaustively enumerates
all paths in $\Psi[n\rightarrow i]$ as we consider all $\psi\in\Psi[(n-1)\rightarrow i]$.
Path $\psi^{'}$ has the same weight as $\psi$ since they have the
same length, while path $\psi^{''}$ has $q$ fraction of the weight
of $\psi$ since it is longer by 1. This shows that the weight of
all paths in $\Psi[n\rightarrow i]$ is equal to $1+q$ times the
weight of all paths in $\Psi[(n-1)\rightarrow i]$; hence, $b_{n,i}=(1+q)b_{n-1,i}$.
\begin{rem}
The case $q=1$ is studied in \citet{eyster2010naive}, who use a
slightly different signal structure. In their setting, \citet{eyster2010naive}
show that agents' beliefs converge to $0$ or $1$ almost surely and
derive a nonzero lower bound on the probability of converging to the
incorrect belief. By contrast, our result gives the exact probability
of converging to the wrong belief for any $0<q\le1$, under a Gaussian
signal structure.
\begin{rem}
\label{rem:discontinuous}There is a discontinuity at $q=0$. As $q$
approaches $0$, the probability of society eventually learning correctly
approaches $1$. But when $q=0$, each agent uses only her own private
signal, so there is no social learning. This nonconvergence of actions
means that society never learns correctly. 
\end{rem}
\end{rem}
While we have focused on long-run learning accuracy, there is a trade-off
between the speed of convergence and asymptotic accuracy for naive
agents. The next proposition illustrates an extreme form of this trade-off.
Start with a uniform-weights network with any link weight $0<q^{*}\le1$.
Sufficiently sparse uniform-weights networks will have worse accuracy
than the $q^{*}$-uniform network for arbitrarily many early agents
due to a lack of information aggregation. However, as implied by Proposition
\ref{prop:Uniform-Weights}, late enough agents will have higher accuracy
on these very sparse networks than on the $q^{*}$-uniform network.
\begin{prop}
\label{prop:very_sparse}For any $0<q^{*}\le1$ and $N\in\mathbb{N},$
there exists some $\bar{q}\in(0,q^{*})$ so that $\mathbb{P}[\tilde{a}_{n}>0\mid\omega=1]$
is strictly larger on the $q^{*}$-uniform weights network than the
$q$-uniform weights network for all $2\le n\le N$ and $q\in(0,\bar{q}).$
\end{prop}

\subsection{Two Groups\label{subsec:Two-Groups}}

We next consider a network with two groups and different weights for
links within groups and between groups. By varying the link weights,
we will consider how homophily (i.e., segregation in communication)
changes learning outcomes.

Odd-numbered agents are in one group and even-numbered agents are
in a second group. Each feasible within-group link has weight $q_{s}$
($s$ for same) and each feasible between-group link has weight $q_{d}$
($d$ for different), so that for $i>j$, the link $M_{i,j}=q_{s}$
if $i\equiv j\ (\text{mod }2)$ and $M_{i,j}=q_{d}$ otherwise. Figure
\ref{fig:two_groups_network_4} illustrates the first four agents
in a two-group network.

\begin{figure}[H]
\begin{centering}
\includegraphics[scale=0.3]{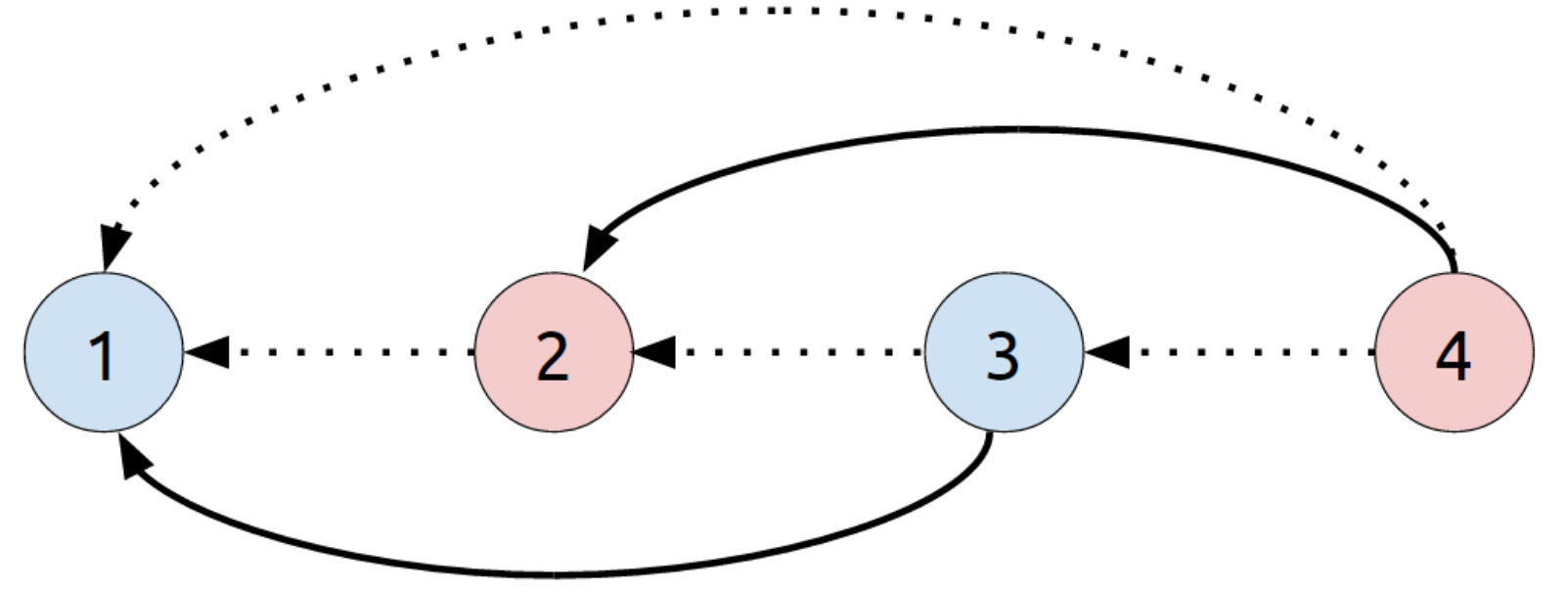}
\par\end{centering}
\caption{\label{fig:two_groups_network_4}First four agents in a two-groups
network. Odd-numbered agents are in one group while even-numbered
agents are in another group. Solid arrows have weight $q_{s}$ and
dashed arrows have weight $q_{d}$.}
\end{figure}

We denote the probability of mislearning with weights $q_{s}$ and
$q_{d}$ as $\xi(q_{s},q_{d})$.
\begin{prop}
\label{prop:twogroups}Consider the two-groups network with within-group
link weight $q_{s}$ and across-group link weight $q_{d}.$ When $0\le q_{s}\le1$
and $0<q_{d}\le1$, almost surely agents' actions $a_{n}$ converge
to $0$ or $1$. The partial derivatives of the mislearning probability
$\xi(q_{s},q_{d})$ satisfy 
\[
\frac{\partial\xi}{\partial q_{d}}>\frac{\partial\xi}{\partial q_{s}}>0,
\]
 i.e., the probability is increasing in $q_{s}$ and $q_{d}$, but
increasing $q_{d}$ has a larger effect than increasing $q_{s}$.
\end{prop}
The first statement again says that agents eventually agree on the
state and eventually have arbitrarily strong beliefs. The fact that
$\xi$ is increasing in $q_{s}$ and $q_{d}$ is another example of
higher link density implying more mislearning. The comparison $\partial\xi/\partial q_{d}>\partial\xi/\partial q_{s}$
tells us that more integrated (i.e., less homophilic) networks are
more likely to herd on the wrong state of the world.

Convergence of beliefs is more subtle with two groups, as we might
imagine the two homophilic groups holding different beliefs asymptotically.
This does not happen because agents have continuous actions that allow
them to precisely convey the strength of their beliefs. As such, eventually
one group will develop sufficiently strong beliefs to convince the
other given any arbitrarily weak connection $q_{d}>0$ between groups.
(In Section \ref{subsec:Disagreement}, however, we show that disagreement
between two homophilic groups is possible with a coarser action space.)

To see that convergence must occur, observe that the belief of a later
agent $n$ depends mostly on the number of paths from that agent to
early agents (and those agents' signal realizations). When $n$ is
large, most paths from agent $n$ to an early agent pass between the
two groups many times. So the number of paths does not depend substantially
on agent $n$'s group. Put another way, when $q_{s}\gg q_{d}>0$ and
$n$ is large, agent $n$ has many more length-1 paths to her own
group than to the other group, but roughly the same total number of
paths across all lengths to both groups. Therefore, agent $n$'s belief
does not depend substantially on whether $n$ is in the odd group
or the even group.\footnote{Each path transitions between the two groups, and eventually the probability
of ending in a given group is approximately independent of the starting
group. This is analogous to a Markov chain approaching its stationary
distribution.}

\citet{coleman1958relational}'s homophily index equals $(q_{s}-q_{d})/(q_{s}+q_{d})$
for this weighted network. To explore how homophily affects mislearning
probability while holding fixed the average degree of each agent,
we consider the total derivative $\frac{d}{d\Delta}\xi(q_{s}+\Delta,q_{d}-\Delta)$.
To interpret, we are considering the marginal effect on mislearning
of a $\Delta$ increase to all the within-group link weights, coupled
with a $\Delta$ decrease to all the between-groups link weights.
These two perturbations, applied simultaneously, leave each agent
with roughly the same total degree and increases the homophily index
by $(2\Delta)/(q_{s}+q_{d})$. Using the chain rule and Proposition
\ref{prop:twogroups},
\[
\frac{d}{d\Delta}\xi(q_{s}+\Delta,q_{d}-\Delta)=\frac{\partial\xi}{\partial q_{s}}-\frac{\partial\xi}{\partial q_{d}}<0,
\]
which means increasing the homophily index of the society and fixing
average degrees always decreases the probability of mislearning. Note
that this result holds regardless of whether society is currently
homophilic ($q_{s}>q_{d})$ or heterophilic ($q_{s}<q_{d})$.

An important insight from the literature about social learning on
networks is that beliefs converge more slowly on more segregated networks
\citep{golub2012homophily}. In our model, faster convergence of beliefs
tends to imply a higher probability of incorrect beliefs. When beliefs
converge quickly, agents are putting far too much weight on early
movers, while when beliefs converge more slowly agents wait for more
independent information. Since agents eventually agree, segregation
helps society form strong beliefs more gradually.

\section{Disagreement\label{subsec:Disagreement}}

In Section \ref{subsec:Two-Groups}, we saw that even on partially
segregated networks, agents eventually reach a consensus on the state
of the world. This agreement relies crucially on the richness of the
action space available to agents, which allows agents to communicate
the strength of their beliefs. In this section, we modify our model
so that the action space is binary and show that the two groups can
disagree forever about the state of the world even when the number
of connections across the groups is unbounded.

The contrasting results for the binary-actions model versus the continuum-actions
model echo a similar contrast in the rational-herding literature,
where society herds on the wrong action with positive probability
when actions coarsely reflect beliefs \citep*{banerjee1992simple,bikhchandani1992theory},
but almost surely converges to the correct action when the action
set is rich enough \citep{lee1993convergence}. Interestingly, while
the rational-herding literature finds that an unboundedly informative
signal structure prevents herding on the wrong action even when actions
coarsely reflect beliefs \citep{smith2000pathological}, we will show
below that even with Gaussian signals two groups may disagree with
positive probability.

Suppose that the state of the world and the signal structure are the
same as in Section \ref{sec:Model}, but agents now choose binary
actions $a_{i}\in\left\{ 0,1\right\} $. Agents still maximize the
expectation of $u_{i}(a_{i},\omega)\coloneqq-(a_{i}-\omega)^{2}$
given their beliefs about $\omega$, under the psychological errors
given by Assumptions \ref{assu:behavioral} and \ref{assu:wrong_precision}.
This utility function now implies that an agent chooses the action
corresponding to the state of world she believes is more likely. Agents
live on the two-groups network from Section \ref{subsec:Two-Groups}:
for $i>j$, the link $M_{i,j}=q_{s}$ if $i\equiv j\ (\text{mod }2)$
and $M_{i,j}=q_{d}$ otherwise. We assume $q_{s}>q_{d}>0$, so that
agents have stronger connections with predecessors from their own
groups.
\begin{prop}
\label{prop:disagreement} Consider the two-groups network. Suppose
$q_{s}>q_{d}>0$ and agents play binary actions. Then there is a positive
probability that all odd-numbered agents choose action $0$ while
all even-numbered agents choose action $1.$
\end{prop}
Persistent disagreement is sustained even though agent $n$ has approximately
$nq_{d}/2$ weighted links to agents from the other group (when $n$
is large) taking opposite actions.

Our result extends to two groups of unequal sizes as long as for all
later agents, the total number of weighted links to their own group
is larger than the total number of weighted links to the other group.

We also get the same result on a random-network analog of the two-groups
model, where edges are unweighted and $q_{s}$ is the probability
of link formation within groups while $q_{d}$ is the probability
of link formation between groups. Agents observe only the actions
of the predecessors to whom they are linked and wrongly believe all
observed actions derive from private signals.\footnote{Details of the statement and a proof are available in a previous draft
at \url{https://arxiv.org/pdf/1703.02105v5.pdf}.} By contrast, with rational agents, Theorem 2 of \citet*{acemoglu2011bayesian}
implies the groups almost surely agree asymptotically on this random
network.

This result adds a new mechanism to the literature on disagreement
in connected societies \citep*{acemoglu2013opinion,acemoglu2016fragility,sethi2012public}.
\citet{bohren2017bounded} also study disagreement in a binary sequential-learning
setting with behavioral agents, but their results concern disagreement
on a complete network among agents with different types of behavioral
biases. By contrast, our Proposition \ref{prop:disagreement} says
that when all agents use the same naive heuristic, they can still
disagree by virtue of belonging to two different homophilic social
groups, even when there are many connections between those groups.

\section{\label{sec:Conclusion}Conclusion}

In this paper, we have explored the influence of network structures
on learning outcomes when agents move sequentially and use a log-linear
learning rule due to inferential naiveté. We have compared long-run
welfare across networks by deriving the exact probabilities of mislearning
on arbitrary networks.

We have studied the simplest possible social-learning environment
to focus on the effect of network structure, but several extensions
are straightforward. Analogs of our general results hold for finite
state spaces with more than two elements, where we can define a log-likelihood
ratio for each pair of states. We can also make the order of moves
random and unknown, in which case naive behavior conditional on a
given turn order is the same as when that order is certain.

We prove our comparative statics results for weighted networks as
they are analytically more tractable than random graphs. For each
weighted network with weights in $[0,1],$ there corresponds an analogous
(unweighted) random graph model where the $i\to j$ link exists with
probability $M_{i,j}$. In numerical simulations, all comparative
statics results proved for weighted networks in Section \ref{sec:Applications}
continue to hold in the analogous random network models. The major
obstacle to extending our proofs is that because our networks are
directed and acyclic, the relevant adjacency matrices have no nonzero
eigenvalues. As a consequence, most techniques from spectral random
graph theory do not apply (but perhaps other methods would).

\newpage{}

\appendix
\begin{center}
\textbf{\Large{}Appendix}{\Large\par}
\par\end{center}

\section{\label{sec:Omitted-Proofs}Proofs}

\subsection{Proof of Lemma \ref{lem:On-an-unweighted}}
\begin{proof}
The log-likelihood ratio of state $\omega=1$ and state $\omega=0$
conditional on the signal realizations of $i$'s linked predecessors
is:

\begin{eqnarray*}
\ln\left(\frac{\mathbb{P}[\omega=1|s_{i},(s_{j})_{j:M_{i,j}=1}]}{\mathbb{P}[\omega=0|s_{i},(s_{j})_{j:M_{i,j}=1}]}\right) & = & \ln\left(\frac{\mathbb{P}[s_{i},(s_{j})_{j:M_{i,j}=1}|\omega=1]}{\mathbb{P}[s_{i},(s_{j})_{j:M_{i,j}=1}|\omega=0]}\right)\text{ (two states equally likely)}\\
 & = & \ln\left(\frac{\mathbb{P}[s_{i}|\omega=1]}{\mathbb{P}[s_{i}|\omega=0]}\cdot\prod_{j:M_{i,j}=1}\frac{\mathbb{P}[s_{j}|\omega=1]}{\mathbb{P}[s_{j}|\omega=0]}\right)\text{ (by independence)}\\
 & = & \ln\left(\frac{\mathbb{P}[\omega=1|s_{i}]}{\mathbb{P}[\omega=0|s_{i}]}\right)+\sum_{j:M_{i,j}=1}\ln\left(\frac{\mathbb{P}[\omega=1|s_{j}]}{\mathbb{P}[\omega=0|s_{j}]}\right)\\
 & = & \tilde{s}_{i}+\sum_{j:M_{i,j}=1}\tilde{s}_{j}\\
 & = & \tilde{s}_{i}+\sum_{j<i}M_{i,j}\tilde{s}_{j}
\end{eqnarray*}

Due to Assumption \ref{assu:behavioral}, $i$ thinks each predecessor
$j$ must have received signal $s_{j}$ such that $\tilde{s}_{j}=\tilde{a}_{j}$.
When $i$ observes only the play of linked predecessors, her log-likelihood
ratio of state $\omega=1$ and state $\omega=0$ given her social
observations and private signal is therefore $\tilde{s}_{i}+\sum_{j<i}M_{i,j}\tilde{a}_{j}$.
She maximizes her expected payoff by choosing an action $a_{i}$ corresponding
to her belief in state $\omega=1,$ which implies that $\tilde{a}_{i}$
is equal to this log-likelihood ratio.
\end{proof}

\subsection{Proof of Lemma \ref{lem:precision}\protect 
}We first establish an auxiliary lemma.
\begin{lem}
\label{lem:llh_normal} We have $\tilde{s}_{i}=2s_{i}/\sigma^{2}$.
\end{lem}
\begin{proof}
The log-likelihood ratio is 
\begin{align*}
\ln\left(\frac{\mathbb{P}[\omega=1|s_{i}]}{\mathbb{P}[\omega=0|s_{i}]}\right) & =\ln\left(\frac{\mathbb{P}[s_{i}|\omega=1]}{\mathbb{P}[s_{i}|\omega=0]}\right)=\ln\left(\frac{\exp\left(\frac{-(s_{i}-1)^{2}}{2\sigma^{2}}\right)}{\exp\left(\frac{-(s_{i}+1)^{2}}{2\sigma^{2}}\right)}\right)\\
 & =\frac{-(s_{i}^{2}-2s_{i}+1)+(s_{i}^{2}+2s_{i}+1)}{2\sigma^{2}}=2s_{i}/\sigma^{2}.\qedhere
\end{align*}
\end{proof}
We now turn to the proof of Lemma \ref{lem:precision}.
\begin{proof}
Due to Assumptions \ref{assu:behavioral}, $i$ thinks that $j$ will
choose $a_{j}$ such that $\tilde{a}_{j}=2s_{j}/\sigma^{2}$ by the
result just established, since $j$ thinks the conditional variance
of her signal is $\sigma^{2}$. But, since $i$ believes $j$'s signal
has conditional variance $\sigma^{2}/M_{i,j}$ by Assumption \ref{lem:precision},
in $i$'s view 
\[
\ln\left(\frac{\mathbb{P}[\omega=1|s_{j}]}{\mathbb{P}[\omega=0|s_{j}]}\right)=\frac{2s_{j}}{\sigma^{2}/M_{i,j}}=M_{i,j}\tilde{a}_{j},
\]
again applying the result above.

Omitting analogous algebraic arguments as in the proof of Lemma \ref{lem:On-an-unweighted},
\begin{align*}
\ln\left(\frac{\mathbb{P}[\omega=1|s_{i},(s_{j})_{j<i}]}{\mathbb{P}[\omega=0|s_{i},(s_{j})_{j<i}]}\right) & =\ln\left(\frac{\mathbb{P}[\omega=1|s_{i}]}{\mathbb{P}[\omega=0|s_{i}]}\right)+\sum_{j<i}\ln\left(\frac{\mathbb{P}[\omega=1|s_{j}]}{\mathbb{P}[\omega=0|s_{j}]}\right)\\
 & =\tilde{s}_{i}+\sum_{j<i}M_{i,j}\tilde{a}_{j}.
\end{align*}
So $\tilde{s}_{i}+\sum_{j<i}M_{i,j}\tilde{a}_{j}$ is $i$'s log-likelihood
ratio of state $\omega=1$ and state $\omega=0$ given her social
observations and private signal.
\end{proof}

\subsection{Proof of Proposition \ref{prop:representation}}
\begin{proof}
By weighted log-linear inference, for each $i$ we have $\tilde{a}_{i}=\tilde{s}_{i}+\sum_{j\in N_{i}}M_{i,j}\tilde{a}_{j}$.
In vector notation, we therefore have 
\[
\left(\begin{array}{c}
\tilde{a}_{1}\\
\vdots\\
\tilde{a}_{n}
\end{array}\right)=\left(\begin{array}{c}
\tilde{s}_{1}\\
\vdots\\
\tilde{s}_{n}
\end{array}\right)+M[n]\cdot\left(\begin{array}{c}
\tilde{a}_{1}\\
\vdots\\
\tilde{a}_{n}
\end{array}\right)
\]

Algebra then yields the desired expression. Note that $(I-M[n])$
is invertible because $M[n]$ is lower triangular with all diagonal
entries equal to 0.

To see the path-counting interpretation, write $(I-M[n])^{-1}=\sum_{k=0}^{\infty}M[n]^{k}$.
Here, $(M[n]^{k})_{i,j}$ counts the number of weighted paths of length
$k$ from $i$ to $j$.
\end{proof}

\subsection{Proof of Proposition \ref{prop:influence_correct_learning}}
\begin{proof}
Without loss of generality, assume $\omega=1$. (The case of $\omega=0$
is exactly analogous and is omitted.) Note that $a_{n}$ converges
in probability to 1 if and only if $\tilde{a}_{n}$ converges in probability
to $\infty$.

First suppose that $\lim_{n\rightarrow\infty}\mathbb{I}(j\to n)\neq0$
for some $j$. Then there exists $\epsilon>0$ such that $\mathbb{I}(j\to n)>\epsilon$
for infinitely many $n$. For each such $n$, the probability that
agent $n$ chooses an action with $\tilde{a}_{n}<0$ is equal to the
probability that $\sum_{i=1}^{n}\mathbb{I}(i\to n)\tilde{s}_{i}$
is negative.

Because $s$ is Gaussian, $\tilde{s}$ has finite variance, so we
can find positive constants $C$ and $\delta$ independent of $n$
such that $\sum_{i\neq j}\mathbb{I}(i\to n)\tilde{s}_{i}<C$ with
probability at least $\delta$ (for example, by applying Markov's
inequality to $|\tilde{s}_{i}|$). Then agent $n$ will be wrong if
$\tilde{s}_{j}<-C/\epsilon$, which is a positive probability event
since $\tilde{s}$ is unbounded. So the probability that an agent
$n$ such that $\mathbb{I}(j\to n)>\epsilon$ chooses $\tilde{a}_{n}<0$
is bounded from below by a positive constant.

For the converse, suppose that $\lim_{n\rightarrow\infty}\mathbb{I}(i\to n)=0$
for all $i$. By the independence of the log signals $\tilde{s}_{i}$,
the log action $\tilde{a}_{n}=\sum_{i=1}^{n}\mathbb{I}(i\to n)\tilde{s}_{i}$
is a random variable with mean $\|\vec{b}_{n}\|_{1}$ and standard
deviation $\|\vec{b}_{n}\|_{2}\sigma$ when $\omega=1$. We now use
the connectedness condition to show that $\|\vec{b}_{n}\|_{1}/\|\vec{b}_{n}\|_{2}\to\infty$.

Find $N$ and $C\le1$ as in the connectedness condition. For each
$\epsilon>0$, we can choose $M_{\epsilon}$ such that $\mathbb{I}(i\to n)<\epsilon$
whenever $i<N$ and $n>M_{\epsilon}$ by the hypothesis $\lim_{n\rightarrow\infty}\mathbb{I}(i\to n)=0$
applied to the finitely many members $i<N$. But now for any $j\ge N$
and any $n>\max(j,M_{\epsilon}),$ concatenating a path from $j$
to $n$ with a path from $i$ to $j$ gives a path from $i$ to $n$
whose weight is the product of the weights of the two subpaths. This
shows $b_{n,i}\ge b_{j,i}\cdot b_{n,j}$, which implies $\mathbb{I}(i\to n)\ge\mathbb{I}(j\to n)\cdot b_{j,i}$.
We have $\mathbb{I}(j\to n)\le\min_{i<N}\mathbb{I}(i\to n)/b_{j,i}$,
where $b_{j,i}\ge C$ for at least one $i<N$ by the connectedness
condition. This shows for any $j\in\mathbb{N}$ and for $n>M_{\epsilon}$,
that we get $\mathbb{I}(j\to n)\le\epsilon/C$.

We have for all $n>M_{\epsilon}$, 
\[
\frac{\|\vec{b}_{n}\|_{2}}{\|\vec{b}_{n}\|_{1}}\leq\max_{j}\frac{\sqrt{\|\vec{b}_{n}\|_{1}\cdot b_{n,j}}}{\|\vec{b}_{n}\|_{1}}=\max_{j<n}\sqrt{\mathbb{I}(j\to n)}<\sqrt{\epsilon/C}.
\]

Because $\epsilon>0$ is arbitrary, $\|\vec{b}_{n}\|_{1}/\|\vec{b}_{n}\|_{2}$
converges to infinity.

Let some $K>0$ be given. We now show that $\mathbb{P}[\tilde{a}_{n}<K|\omega=1]\to0$,
hence proving that $\tilde{a}_{n}$ converges to $\infty$ in probability.
We compute 
\[
z_{n}:=\frac{\mathbb{E}[\tilde{a}_{n}|\omega=1]-K}{\text{Std}[\tilde{a}_{n}|\omega=1]}=\frac{\|\vec{b}_{n}\|_{1}\cdot\frac{2}{\sigma^{2}}}{\|\vec{b}_{n}\|_{2}\cdot\frac{2}{\sigma}}-\frac{K}{\|\vec{b}_{n}\|_{2}\cdot\frac{2}{\sigma}}.
\]

Since $\|\vec{b}_{n}\|_{1}/\|\vec{b}_{n}\|_{2}\to\infty$, the first
term converges to infinity. By the connectedness condition, $\|\vec{b}_{n}\|_{2}\ge C$
for all large enough $n$, so the second term is bounded. This implies
$z_{n}\to\infty$. By Chebyshev's inequality, $\mathbb{P}[\tilde{a}_{n}<K|\omega=1]\le z_{n}^{-2}$.
This shows $\mathbb{P}[\tilde{a}_{n}<K|\omega=1]\to0$.

We note that this shows convergence in probability, but does not characterize
the joint distribution of actions, so these methods do not guarantee
almost sure convergence (without further structure on the networks
as in Section \ref{sec:Applications}).
\end{proof}

\subsection{Proof of Proposition \ref{prop:one_plus_epsilon}}
\begin{proof}
By the hypothesis of the proposition, there exists some $\epsilon>0$
and $N\in\mathbb{N}$ so that for all $i>N,$ $\sum_{j<i}M_{i,j}\ge1+\epsilon$.
Modify the network to set all links originating from any of the first
$N$ agents to have weight 0, that is, $M_{i,j}=0$ for all $i,j\le N$.

We prove by induction that $\sum_{j\le N}b_{i,j}\ge1+\epsilon$ for
all $i\ge N+1$ on the modified network. Consider agent $N+1$. Since
$\sum_{j<N+1}M_{N+1,j}\ge1+\epsilon$ and all of $(N+1)$'s out-degree
comes from links to agents in position $N$ or earlier, $\sum_{j\le N}b_{N+1,j}\ge1+\epsilon$.
By induction, suppose $\sum_{j\le N}b_{N+k,j}\ge1+\epsilon$ holds
for all $1\le k\le K$. A lower bound on $\sum_{j\le N}b_{N+K+1,j}$
is 
\begin{align*}
\sum_{j\le N}M_{N+K+1,j}+\sum_{N+1\le i\le N+K}M_{N+K+1,i}\cdot\left(\sum_{j^{'}\le N}b_{i,j'}\right) & \ge\sum_{j\le N}M_{N+K+1,j}+\sum_{N+1\le i\le N+K}M_{N+K+1,i}\cdot(1+\epsilon)\\
 & \ge\sum_{j\le N+K}M_{N+K+1,j}\\
 & \ge1+\epsilon,
\end{align*}
where in the first inequality we used the inductive hypothesis and
in the last inequality we used the fact that $N+K+1$ has an out-degree
of at least $1+\epsilon$. This establishes $\sum_{j\le N}b_{N+K+1,j}\ge1+\epsilon$
and so by induction, $\sum_{j\le N}b_{N+k,j}\ge1+\epsilon$ for all
$k\ge1.$

This result holds a fortiori on the original network with higher link
weights. By the pigeonhole principle, the original network satisfies
the connectedness condition with $C=\frac{1}{N}(1+\epsilon)$.

Now return to the modified network (so $M$ refers to the possibly
modified network weights). We develop some notation for the rest of
the proof and establish an intermediary lemma. For $j<i,$ let $\Psi[i\to j]$
be the set of all paths from $i$ to $j$. Let $\hat{\Psi}[i\to[N]]$
be the set of paths from $i$ to some agent $k\le N$, such that the
path contains no links between two different agents among the first
$N.$ Let $\hat{\Psi}[i\to[N]\mid j]$ be the subset of such paths
that pass through $j$.

For a path $\psi$ passing through agents $i_{1},i_{2},...,i_{L}$,
let $W(\psi):=\prod_{\ell=1}^{L-1}M_{i_{\ell+1},i_{\ell}}$ denote
its weight and let $D(\psi):=\prod_{\ell=1}^{L-1}(\sum_{j<i_{\ell}}M_{i_{\ell},j})$
denote the product of out-degrees of all agents on the path except
the last one.
\begin{lem}
\label{lem:splitting} For $n>N$, $\sum_{\psi\in\hat{\Psi}[n\to[N]]}\frac{W(\psi)}{D(\psi)}=1$.
\end{lem}
\begin{proof}
We prove by induction on $n.$ For $n=N+1,$ the set $\hat{\Psi}[n\to[N]]$
is the set of $N$ paths each consisting of a link from $N+1$ to
some agent $j\le N$. Each $\psi\in\hat{\Psi}[n\to[N]]$ therefore
has $D(\psi)=\sum_{j<N+1}M_{N+1,j}$, and the path terminating at
$j$ has $W(\psi)=M_{N+1,j}$. So the claim holds for $n=N+1.$ By
induction suppose it holds for all $n\le N+K$ for some $K\ge1.$
For $n=N+K+1,$ partition $\hat{\Psi}[n\to[N]]$ into $K+1$ groups.
For $1\le k\le K,$ each path $\psi\in\Psi(k)$ in the $k$th group
consists of the link $n\to(N+k)$ concatenated in front of a path
$\psi^{'}\in\hat{\Psi}[(N+k)\to[N]]$, so $\psi=((n,N+k),\psi^{'})$.
The final $(K+1$)th group consists of paths where $n$ links directly
to an agent among the first $N$. We have: 
\begin{align*}
\sum_{\psi\in\hat{\Psi}[n\to[N]]}\frac{W(\psi)}{D(\psi)} & =\sum_{k=1}^{K}\left(\sum_{\psi^{'}\in\hat{\Psi}[(N+k)\to[N]]}\frac{W((n,N+k),\psi^{'})}{D((n,N+k),\psi')}\right)+\sum_{j=1}^{N}\frac{M_{n,j}}{\sum_{h<n}M_{n,h}}\\
 & =\sum_{k=1}^{K}\left(\sum_{\psi^{'}\in\hat{\Psi}[(N+k)\to[N]]}\frac{M_{n,N+k}\cdot W(\psi^{'})}{\sum_{h<n}M_{n,h}\cdot D(\psi')}\right)+\sum_{j=1}^{N}\frac{M_{n,j}}{\sum_{h<n}M_{n,h}}\\
 & =\sum_{k=1}^{K}\left(\frac{M_{n,N+k}}{\sum_{h<n}M_{n,h}}\cdot1\right)+\sum_{j=1}^{N}\frac{M_{n,j}}{\sum_{h<n}M_{n,h}}\text{ (by inductive hypothesis)}\\
 & =\frac{\sum_{h<n}M_{n,h}}{\sum_{h<n}M_{n,h}}=1.
\end{align*}

So by induction, this claim holds for all $n>N.$
\end{proof}
We now return to the proof of Proposition \ref{prop:one_plus_epsilon}.
For $N<i<n$, 
\begin{align*}
b_{n,i} & =\sum_{\psi\in\Psi[n\to i]}W(\psi)=\sum_{\psi\in\Psi[n\to i]}\left[W(\psi)\cdot\left(\sum_{\hat{\psi}\in\hat{\Psi}[i\to[N]]}\frac{W(\hat{\psi})}{D(\hat{\psi})}\right)\right],
\end{align*}
 where the second equality follows because Lemma \ref{lem:splitting}
implies the term in the inner parentheses is 1. For a path $\psi$
passing through $i$, let $\psi[i]$ denote the subpath starting with
$i.$ So the above says 
\[
b_{n,i}=\sum_{\psi\in\hat{\Psi}[n\to[N]\mid i]}\frac{W(\psi)}{D(\psi[i])}.
\]
 Summing across $i$, we may re-index the sum by paths in $\hat{\Psi}[n\to[N]]$.
To be more precise, for $\psi\in\hat{\Psi}[n\to[N]]$, write $A(\psi)\subseteq\{N+1,...,n-1\}$
to be the set of agents that $\psi$ passes through. For each $j\in A(\psi),$
we have $\psi\in\hat{\Psi}[n\to[N]\mid j],$ so it contributes $W(\psi)/D(\psi[j])$
to the overall sum, 
\begin{align*}
\sum_{i=N+1}^{n-1}b_{n,i} & =\sum_{i=N+1}^{n-1}\sum_{\psi\in\hat{\Psi}[n\to[N]\mid i]}\frac{W(\psi)}{D(\psi[i])}\\
 & =\sum_{\psi\in\hat{\Psi}[n\to[N]]}\sum_{j\in A(\psi)}\frac{W(\psi)}{D(\psi[j])}\\
 & \le\sum_{\psi\in\hat{\Psi}[n\to[N]]}W(\psi)\cdot\sum_{j\in A(\psi)}\frac{1}{(1+\epsilon)^{|\psi[j]|-1}}\\
 & \le\sum_{\psi\in\hat{\Psi}[n\to[N]]}W(\psi)\cdot\frac{1}{\epsilon}
\end{align*}
 where $|\psi[j]|$ denotes the number of agents in the subpath $\psi[j]$.
On the third line, we used the fact that all agents on $\psi$ except
the last one must have out-degree at least $1+\epsilon,$ so $D(\psi[j])\ge(1+\epsilon)^{|\psi[j]|-1}$.
The result $\sum_{i=N+1}^{n-1}b_{n,i}\le\sum_{\psi\in\hat{\Psi}[n\to[N]]}W(\psi)\cdot\frac{1}{\epsilon}$
also holds for the original network, since we have not modified the
subnetwork among agents $N+1,...,n$.

We also have $\sum_{i=1}^{N}b_{n,i}=\sum_{\psi\in\hat{\Psi}[n\to[N]]}W(\psi)$.
On the original network, we have higher link weights among the first
$N$ agents, so we in fact have 
\[
\sum_{i=1}^{N}b_{n,i}\ge\sum_{\psi\in\hat{\Psi}[n\to[N]]}W(\psi).
\]
So, on the original network, 
\[
\sum_{i=1}^{N}\mathbb{I}(i\to n)\ge\frac{1}{1+1/\epsilon}.
\]
This inequality holds for every $n$, so it cannot be the case that
$\lim_{n\to\infty}\mathbb{I}(i\to n)=0$ for all $1\le i\le N$.
\end{proof}

\subsection{Proof of Example \ref{exa:decay}}
\begin{proof}
The coefficients $b_{i,n}$ satisfy the recurrence relation $b_{n,i}=2\delta b_{n-1,i}$
whenever $n-i>1$.

When $\delta=\frac{1}{2}$, from the recurrence relation, all predecessors'
signals are given equal weight, so by the law of large numbers, actions
converge to $\omega$ almost surely.

When $\delta>\frac{1}{2}$, $\sum_{k=0}^{\infty}\delta^{k}>1,$ so
$\sum_{j<i}M_{i,j}$ is bounded above 1 for $n$ large enough. So
by Proposition \ref{prop:one_plus_epsilon}, society does not learn
correctly.

The final case is $\delta<\frac{1}{2}$. We show that $\mathbb{P}[a_{n}\le\frac{1}{2}\mid\omega=1]$
is bounded away from 0 for all $n\ge1$, so $(a_{n})$ cannot converge
in probability to $\omega$.

Without loss of generality, normalize to $\sigma=1.$ From the recurrence
relation for the coefficients $b_{n,i},$ it is easy to check that
$\tilde{a}_{n}=2\delta\tilde{a}_{n-1}+\tilde{s}_{n}-\delta\tilde{s}_{n-1}$
for each $n.$ Evidently $b_{i+1,i}=\delta,$ so from recursion, $b_{n,i}=(2\delta)^{n-i-1}\delta$
for $i\le n-1$, $b_{n,n}=1$. So, $\tilde{a}_{n}=\tilde{s}_{n}+\delta\sum_{j=0}^{n-2}(2\delta)^{j}\cdot\tilde{s}_{n-1-j},$
meaning 
\[
\tilde{a}_{n}\mid(\omega=1)\sim\mathcal{N}\left(1+\delta\frac{1-(2\delta)^{n-1}}{1-2\delta},1+\delta^{2}\frac{1-(4\delta^{2})^{n-1}}{1-4\delta^{2}}\right).
\]
So 
\[
\mathbb{P}[\tilde{a}_{n}\le0\mid\omega=1]\ge\Phi\left(-1-\frac{\delta}{1-2\delta}\right)
\]
 for all $n$, which implies for all $n$, 
\[
\mathbb{P}[a_{n}\le\frac{1}{2}\mid\omega=1]\ge\Phi\left(-1-\frac{\delta}{1-2\delta}\right).
\]
\end{proof}

\subsection{Proof of Proposition \ref{prop:prob_mislearn}}

We first state and prove a lemma that gives the ex ante distribution
of agent $n$'s log action.
\begin{lem}
\label{lem:prob_correct}When $\omega=1$, the log action of agent
$n$ on any weighted network is distributed as 
\[
\tilde{a}_{n}\sim\mathcal{N}\left(\frac{2}{\sigma^{2}}\|\vec{b}_{n}\|_{1},\frac{4}{\sigma^{2}}\|\vec{b}_{n}\|_{2}^{2}\right).
\]
\end{lem}
\begin{proof}
By Proposition \ref{prop:representation}, $\tilde{a}_{n}=\sum_{i=1}^{n}b_{n,i}\tilde{s}_{i}$.
This is equal to $\sum_{i=1}^{n}2b_{n,i}s_{i}/\sigma^{2}$ according
to Lemma \ref{lem:llh_normal}. Conditional on $\omega=1$, $(s_{i})$
are i.i.d$.$ $\mathcal{N}(1,\sigma^{2})$ random variables, so 
\[
\sum_{i=1}^{n}\frac{2b_{n,i}}{\sigma^{2}}s_{i}\sim\mathcal{N}\left(\frac{2}{\sigma^{2}}\sum_{i=1}^{n}b_{n,i},\frac{4}{\sigma^{2}}\sum_{i=1}^{n}b_{n,i}^{2}\right)=\mathcal{N}\left(\frac{2}{\sigma^{2}}\|\vec{b}_{n}\|_{1},\frac{4}{\sigma^{2}}\|\vec{b}_{n}\|_{2}^{2}\right).\qedhere
\]
\end{proof}
Now we give the proof of Proposition \ref{prop:prob_mislearn}.
\begin{proof}
By Lemma \ref{lem:prob_correct}, $\tilde{a}_{n}|(\omega=1)\sim\mathcal{N}\left(\frac{2}{\sigma^{2}}\|\vec{b}_{n}\|_{1},\frac{4}{\sigma^{2}}\|\vec{b}_{n}\|_{2}^{2}\right).$
So using properties of the Gaussian distribution, 
\[
\mathbb{P}[\tilde{a}_{n}>0\mid\omega=1]=\Phi\left(\frac{\frac{2}{\sigma^{2}}\|\vec{b}_{n}\|_{1}}{\frac{2}{\sigma}\|\vec{b}_{n}\|_{2}}\right)=\Phi\left(\frac{1}{\sigma}\cdot\frac{\|\vec{b}_{n}\|_{1}}{\|\vec{b}_{n}\|_{2}}\right).\qedhere
\]
\end{proof}

\subsection{Proof of Proposition \ref{prop:Uniform-Weights}}
\begin{proof}
The numbers of paths from various agents to agent $i$ satisfy the
recurrence relation $b_{n,i}=(1+q)b_{n-1,i}$ when $n-i>1$. By a
simple computation, we find that 
\[
\tilde{a}_{n}=\sum_{i=1}^{n-1}q(1+q)^{n-i-1}\tilde{s}_{i}+\tilde{s}_{n}.
\]
Since $\tilde{s}_{i}$ are independent Gaussian random variables,
our argument uses the fact that for $n$ large, $\tilde{a}_{n}$ has
the same sign as another Gaussian random variable, whose mean and
variance we can compute.

We first show that $\tilde{a}_{n}$ converges to $-\infty$ or $\infty$
almost surely. Consider the random variable 
\[
X_{n}(\vec{s}):=\frac{1}{2}\sum_{i=1}^{n-1}(1+q)^{-i}\tilde{s}_{i},
\]
where $\vec{s}:=(s_{i})_{i=1}^{\infty}$ is the profile of private
signal realizations. By a standard result, $X_{n}(\vec{s})$ converges
almost surely to a random variable $Y(\vec{s})$ such that the conditional
distribution of $Y$ in each state of the world is Gaussian. For each
$n,$ $\tilde{a}_{n}(\vec{s})=2q(1+q)^{n-1}\cdot X_{n}(\vec{s})+\tilde{s}_{n}$.
Since $\sum_{n=1}^{\infty}\mathbb{P}[\tilde{s}_{n}>n]<\infty,$ by
the Borel\textendash Cantelli lemma, $\mathbb{P}[\tilde{s}_{n}>n\text{ infinitely often}]=0$.
So almost surely, $\lim_{n\to\infty}\tilde{a}_{n}(\vec{s})=\lim_{n\to\infty}2q(1+q)^{n-1}\cdot Y(\vec{s})+\tilde{s}_{n}\in\{-\infty,\infty\}$.
This in turn shows that $a_{n}$ converges to 0 or 1 almost surely.

Now we show $\mathbb{P}[a_{n}\to0\mid\omega=1]=\Phi\left(-\sigma^{-1}\sqrt{(q+2)/q}\right),$
which is the same probability as $\mathbb{P}[\tilde{a}_{n}\to-\infty\mid\omega=1]$.
The random variable $Y(\vec{s})$ that $X_{n}(\vec{s})$ converges
to a.s. has the distribution $\mathcal{N}(1/(\sigma^{2}q),1/(\sigma^{2}q(q+2)))$
when $\omega=1$, and $\tilde{a}_{n}$ has the same sign as $X_{n}(\vec{s})$
with probability converging to 1 for $n$ large. The distribution
$\mathcal{N}(1/(\sigma^{2}q),1/(\sigma^{2}q(q+2)))$ assigns $\Phi\left(-\sigma^{-1}\sqrt{(q+2)/q}\right)$
probability to the positive region. The symmetric argument holds for
$\omega=0.$
\end{proof}

\subsection{Proof of Proposition \ref{prop:very_sparse}}

First, we derive a closed-form expression for the probability that
the $n$th agent thinks the correct state is more likely in the uniform
weights network, conditional on $\omega=1$.
\begin{lem}
\label{lem:prob_n_correct_ER} In the $q$-uniform weights network,
\[
\mathbb{P}[\tilde{a}_{n}>0\mid\omega=1]=\Phi\left(\frac{1}{\sigma}\cdot\frac{(1+q)^{n-1}\cdot\sqrt{(2+q)}}{\sqrt{2+q(1+q)^{2n-2}}}\right).
\]
This probability is strictly increasing in $n$ when $0<q\le1.$
\end{lem}
\begin{proof}
From the proof of Proposition \ref{prop:Uniform-Weights}, we have
\[
\tilde{a}_{n}=\sum_{i=1}^{n-1}q(1+q)^{n-i-1}\tilde{s}_{i}+\tilde{s}_{n}
\]
 where the different $\tilde{s}_{i}$'s are conditionally independent
given $\omega=1,$ with $\tilde{s}_{i}\mid(\omega=1)\sim\mathcal{N}\left(2/\sigma^{2},4/\sigma^{2}\right)$
from Lemma \ref{lem:llh_normal}. Thus, the sum $\tilde{a}_{n}$ is
conditionally Gaussian with a mean of 
\begin{align*}
\frac{2}{\sigma^{2}}\cdot\left[1+\sum_{i=1}^{n-1}q(1+q)^{n-i-1}\right] & =\frac{2}{\sigma^{2}}\cdot\left[1+q\cdot\frac{(1+q)^{n-1}-1}{(1+q)-1}\right]\\
 & =\frac{2}{\sigma^{2}}\cdot(1+q)^{n-1}
\end{align*}
 and a variance of 
\begin{align*}
\frac{4}{\sigma^{2}}\cdot\left[1+\sum_{i=1}^{n-1}q^{2}(1+q)^{2n-2i-2}\right] & =\frac{4}{\sigma^{2}}\cdot\left[1+q^{2}\cdot\frac{(1+q)^{2n-2}-1}{(1+q)^{2}-1}\right]\\
 & =\frac{4}{\sigma^{2}}\cdot\frac{2+(1+q)^{2n-2}q}{2+q}.
\end{align*}
 Thus, 0 is 
\[
\frac{\frac{2}{\sigma^{2}}\cdot(1+q)^{n-1}}{\sqrt{\frac{4}{\sigma^{2}}\cdot\frac{2+(1+q)^{2n-2}q}{2+q}}}=\frac{1}{\sigma}\cdot\frac{(1+q)^{n-1}\cdot\sqrt{(2+q)}}{\sqrt{2+q(1+q)^{2n-2}}}
\]
 standard deviations below the mean in the distribution of $\tilde{a}_{n}\mid(\omega=1)$,
so 
\[
\mathbb{P}[\tilde{a}_{n}>0\mid\omega=1]=\Phi\left(\frac{1}{\sigma}\cdot\frac{(1+q)^{n-1}\cdot\sqrt{(2+q)}}{\sqrt{2+q(1+q)^{2n-2}}}\right).
\]
 To see that this expression is strictly increasing, let $n\ge1.$
Then, 
\begin{align*}
\mathbb{P}[\tilde{a}_{n+1}>0\mid\omega=1] & =\Phi\left(\frac{1}{\sigma}\cdot\frac{(1+q)\cdot(1+q)^{n-1}\cdot\sqrt{(2+q)}}{\sqrt{2+(1+q)^{2}\cdot q(1+q)^{2n-2}}}\right)\\
 & >\Phi\left(\frac{1}{\sigma}\cdot\frac{(1+q)\cdot(1+q)^{n-1}\cdot\sqrt{(2+q)}}{\sqrt{(1+q)^{2}\cdot2+(1+q)^{2}\cdot q(1+q)^{2n-2}}}\right)\\
 & =\Phi\left(\frac{1}{\sigma}\cdot\frac{(1+q)^{n-1}\cdot\sqrt{(2+q)}}{\sqrt{2+q(1+q)^{2n-2}}}\right)\\
 & =\mathbb{P}[\tilde{a}_{n}>0\mid\omega=1]
\end{align*}
 as desired.
\end{proof}
Now we give the proof of Proposition \ref{prop:very_sparse}.
\begin{proof}
Let $\mathbb{P}[\tilde{a}_{2}>0\mid\omega=1]$ on the $q^{*}$-uniform
weights network be denoted by $p$. Lemma \ref{lem:prob_n_correct_ER}
implies $\mathbb{P}[\tilde{a}_{n}>0\mid\omega=1]$ is strictly increasing
in $n$ on the $q^{*}$-uniform weights network, so $p>\ensuremath{\Phi(1/\sigma)}$
and, furthermore, $\mathbb{P}[\tilde{a}_{n}>0\mid\omega=1]\ge p$
for all $n\ge2$ on the same network.

The function 
\begin{eqnarray*}
q & \mapsto & \Phi\left(\frac{1}{\sigma}\cdot\frac{(1+q)^{N-1}\cdot\sqrt{(2+q)}}{\sqrt{2+q(1+q)^{2N-2}}}\right)
\end{eqnarray*}
 is continuous and equals $\Phi(1/\sigma)$ when $q=0$. So we may
find a small enough $\bar{q}\in(0,q^{*})$ so that whenever $0<q<\bar{q},$
\[
\Phi\left(\frac{1}{\sigma}\cdot\frac{(1+q)^{N-1}\cdot\sqrt{(2+q)}}{\sqrt{2+q(1+q)^{2N-2}}}\right)<p.
\]
From the monotonicity result of Lemma \ref{lem:prob_n_correct_ER},
this also implies 
\[
\Phi\left(\frac{1}{\sigma}\cdot\frac{(1+q)^{n-1}\cdot\sqrt{(2+q)}}{\sqrt{2+q(1+q)^{2n-2}}}\right)<p
\]
 for all $2\le n\le N$.
\end{proof}

\subsection{Proof of Proposition \ref{prop:twogroups}}
\begin{proof}
Suppose we have two groups, and agents observe predecessors in the
same group with weight $q_{s}$ and predecessors in the other group
with weight $q_{d}$. Then the coefficients $b_{n,i}$ satisfy the
recurrence relation 
\[
b_{n,i}=q_{d}b_{n-1,i}+(1+q_{s})b_{n-2,i}
\]
when $n-i>2$. Since the network is translation invariant, $b_{n,i}$
only depends on $n-i$. By a standard algebraic fact, there exist
constants $c_{+},c_{-},\zeta_{+},\zeta_{-}$ (only depending on $n-i)$
so that 
\[
b_{n,i}=c_{+}\zeta_{+}^{n-i}+c_{-}\zeta_{-}^{n-i},
\]
where $\zeta_{\pm}$ are the solutions to the polynomial $x^{2}-q_{d}x-(1+q_{s})=0$
and $c_{+},c_{-}$ are constants that we can determine from $b_{2,1}$
and $b_{3,1}$. We compute 
\[
\zeta_{\pm}=\frac{q_{d}\pm\sqrt{4q_{s}+q_{d}+4}}{2},
\]
 where $\zeta_{+}>1$ and $\zeta_{-}<0.$ By arguments analogous to
those in the proof of Proposition \ref{prop:Uniform-Weights}, we
may again establish that $a_{n}$ converges to 0 or 1 almost surely.
We now analyze the probability of mislearning.

Since $\zeta_{+}>|\zeta_{-}|$, the exponential term with base $\zeta_{+}$
dominates as $n$ grows large. This shows $c_{+}>0$, since $b_{n,i}$
counts the number of weighted paths in a network so it must be a positive
number. This also shows that $\mathbb{P}[\tilde{a}_{n}<0\mid\omega=1]\to\mathbb{P}[\sum_{i=0}^{\infty}(\zeta_{+})^{-i}\tilde{s}_{i}<0\mid\omega=1]$
as $n\to\infty.$ Conditional on $\omega=1,$ the sum $\sum_{i=0}^{\infty}(\zeta_{+})^{-i}\tilde{s}_{i}$
has the distribution 
\[
\mathcal{N}\left(\frac{2}{\sigma^{2}(\zeta_{+}-1)},\frac{4}{\sigma^{2}(\zeta_{+}-1)(\zeta_{+}+1)}\right),
\]
so it is easy to show that the probability assigned to the negative
region is increasing in $\zeta_{+}$.

Having shown that the probability of mislearning is monotonically
increasing in $\zeta_{+}$, we can take comparative statics: 
\[
\frac{\partial\zeta_{+}}{\partial q_{d}}=\frac{q_{d}}{2\sqrt{4q_{s}+q_{d}+4}}+\frac{1}{2}\text{ and }\frac{\partial\zeta_{+}}{\partial q_{s}}=\frac{1}{\sqrt{4q_{s}+q_{d}+4}}.
\]
It is easy to see that $\partial\zeta_{+}/\partial q_{d}>\partial\zeta_{+}/\partial q_{s}>0$
for all $q_{s}\geq0$ and $q_{d}>0$.
\end{proof}

\subsection{Proof of Proposition \ref{prop:disagreement}}
\begin{proof}
Define $\kappa_{q}$ to be a naive agent $i$'s log-likelihood ratio
of state $\omega=1$ versus state $\omega=0$ upon observing one neighbor
$j$ who picks action 1 with weight $q$. Then we have:
\[
\kappa_{q}\coloneqq\ln\left(\frac{\mathbb{P}[\omega=1|s_{j}\ge0]}{\mathbb{P}[\omega=0|s_{j}<0]}\right)>0,
\]
where $\mathbb{P}$ is taken under $i$'s beliefs about the conditional
distributions of $s_{j}$ under Assumption \ref{assu:wrong_precision},
that is $s_{j}\mid(\omega=1)\sim\mathcal{N}(1,\sigma^{2}/q)$ and
$s_{j}\mid(\omega=0)\sim\mathcal{N}(-1,\sigma^{2}/q)$. In particular,
this log likelihood $\kappa_{q}$ is decreasing in $q$ and so $\kappa_{q_{s}}-\kappa_{q_{d}}>0$
for $q_{s}>q_{d}$.

By symmetry of the Gaussian distribution, the log-likelihood ratio
after observing one neighbor who chooses action 0 with weight $q$
is $-\kappa_{q}$.

Suppose after $2n$ agents have moved, the actions taken so far involve
every odd-numbered agent playing 1and every even-numbered agent playing
0. Then agent $2n+1$ has a log-likelihood ratio of $n(\kappa_{q_{s}}-\kappa_{q_{d}})$
from her social observations. The probability that private signal
$s_{2n+1}$ is so strongly in favor of $\omega=0$ as to make $2n+1$
play $0$ is 
\[
\epsilon_{n}\coloneqq\mathbb{P}\left[s_{i}\in\mathbb{R}:\ln\left(\frac{\mathbb{P}[\omega=1|s_{i}]}{\mathbb{P}[\omega=0|s_{i}]}\right)<-n(\kappa_{q_{s}}-\kappa_{q_{d}})\mid\omega=1\right]
\]
For the Gaussian distribution, $\ln\left(\mathbb{P}[\omega=1|s_{i}]/\mathbb{P}[\omega=0|s_{i}]\right)=2s_{i}/\sigma^{2}$,
so 
\[
\sum_{n=1}^{\infty}\epsilon_{n}=\sum_{n=1}^{\infty}\Phi(-\frac{\sigma^{2}}{2}n(\kappa_{q_{s}}-\kappa_{q_{d}});1,\sigma^{2})<\infty
\]
because the Gaussian distribution function tends to 0 faster than
geometrically. This shows that there is apositive probability that
every odd-numbered agent plays 1.

By an analogous argument, there is also a positive probability that
every even-numbered agent plays 0. In that argument we would use the
fact that 
\[
\sum_{n=1}^{\infty}\left[1-\Phi(\frac{\sigma^{2}}{2}n(\kappa_{q_{s}}-\kappa_{q_{d}});1,\sigma^{2})\right]<\infty.\qedhere
\]
\end{proof}
\bibliographystyle{ecta}
\bibliography{network_naive_sequential}

\end{document}